\documentclass{sig-alternate-05-2015}

\usepackage{mathrsfs}
\newtheorem{Theorem}{Theorem}

\def\scrM{\mathscr{M}}
\def\sfM{\mathsf{M}}

\usepackage{lmodern}
\usepackage[T1]{fontenc}
\usepackage[utf8]{inputenc}
\usepackage{color}

\usepackage[plainpages=false,pdfpagelabels,colorlinks=true,citecolor=blue,hypertexnames=false]{hyperref}

\usepackage{amssymb,amsmath}
\usepackage{mathrsfs}
\usepackage{algorithmicx}
\usepackage{algpseudocode}
\usepackage{algorithm}
\usepackage{breqn}

\newtheorem{Lemma}{Lemma}

\newtheorem{Corollary}{Corollary}
\newtheorem{Definition}{Definition}

\newcommand{\K}{\mathbb{K}}
\newcommand{\R}{\mathcal R_k}

\newcommand{\remx}{~\mathrm{rem}_x~}
\newcommand{\remy}{~\mathrm{rem}_y~}

\renewcommand{\Xi}[1]{[x^i]#1}
\renewcommand{\xi}[1]{#1_i}
\newcommand{\Xv}[1]{[x^{\v(#1)}]#1}

\newcommand{\N}[1]{\widetilde{#1}}

\newcommand{\Seq}{\mathcal{S}}

\newcommand{\Q}{\mathcal{Q}}
\renewcommand{\H}{\mathcal{H}}
\renewcommand{\v}{\upsilon}

\newcommand{\Id}{\mathrm{Id}}
\renewcommand{\O}{\mathcal{O}}

\newcommand{\U}{\mathcal{U}}
\newcommand{\m}{\mathbf{M}}
\newcommand{\s}{s}
\renewcommand{\u}{u}
\renewcommand{\t}{t}

\newcommand{\Res}{\,\mathrm{Res}\,}
\newcommand{\ResIn}{\text{Res}}

\newcommand{\generichalfgcd}{\textsc{GenericHalfGcd}\,}

\newcommand{\mat}[1]{\mathbf{#1}}

\newcommand{\mydivy}{~\mathrm{div}_y~}

\def\Lc{\mathrm{Lc}}

\doi{?}
\isbn{?}

\title{A Fast Algorithm for Computing the Truncated Resultant}

\numberofauthors{2} 
\author{
\alignauthor Guillaume Moroz\\
\affaddr{Inria Nancy Grand Est}\\
  \email{{guillaume.moroz@inria.fr}}
\alignauthor \'Eric Schost\\
  \affaddr{University of Waterloo}\\
  \email{{eschost@uwaterloo.ca}}
}

\date\today

\begin{document}
\maketitle

\begin{abstract}
  Let $P$ and $Q$ be two polynomials in $\K[x,y]$ with degree at most
  $d$, where $\K$ is a field. Denoting by $R\in \K[x]$ the resultant
  of $P$ and $Q$ with respect to $y$, we present an algorithm to
  compute $R \bmod x^k$ in $ \O\tilde{~}(kd)$ arithmetic operations in
  $\K$, where the $\O\tilde{\ }$ notation indicates that we omit
  polylogarithmic factors. This is an improvement over
  state-of-the-art algorithms that require to compute $R$ in
  $\O\tilde{~}(d^3)$ operations before computing its first $k$
  coefficients.
\end{abstract}





\section{Introduction}

Computing the resultant of two polynomials is an ubiquitous question
in symbolic computation, with applications to polynomial system
solving~\cite{GiLeSa01}, computational
topology~\cite{GoKa96,BeEmSa11,EmSa12,BoLaPoRo13,BoLaPoRo13b,BoLaMoPoRo14,KoSa15,BoLaMoPoRoSa15},
Galois theory~\cite{Soicher81,Lehobey97,AuVa00,LeSc12}, computations
with algebraic numbers~\cite{BoFlSaSc06},~etc.

From the complexity viewpoint, this question admits a satisfactory
answer in the simplest case of polynomials $P,Q$ with coefficients in
a field $\K$.  Euclid's algorithm can be adapted to compute the
resultant $R$ of $P$ and $Q$ in time $\O(d^2)$, assuming arithmetic
operations in $\K$ are counted at unit cost (a thorough discussion of
resultant algorithms based on Euclid's algorithm is
in~\cite{GaLu03}). Using fast polynomial multiplication and
divide-and-conquer techniques, the Knuth-Sch\"onhage half-gcd
algorithm~\cite{Knu70,Sch71} allows one to compute $R$ in time
$\O\tilde{~}(d)$. This is optimal, up to logarithmic factors, since
the input has size $\Theta(d)$.

However, no such quasi-linear result is known in the important case of
bivariate polynomials over $\K$, or in the very similar case of
univariate polynomial with integer coefficients (in which case one
would be interested in bit complexity estimates). In the former
situation, suppose we consider two polynomials $P$ and $Q$ in
$\K[x,y]$, with degree at most $d$, and we want to compute their
resultant $R$ with respect to $y$, so that $R$ is in $\K[x]$. The
polynomial $R$ has degree at most $d^2$, so both input and output
can be represented using $\Theta(d^2)$ elements in $\K$. However, the best
known algorithms to compute $R$ take $\O\tilde{~}(d^3)$ operations in
$\K$, either by means of evaluation~/ interpolation techniques, or in
a direct manner~\cite{Reischert97}.

In this paper, we are interested in the computation of the resultant
$R$ of such bivariate polynomials {\em truncated at order $k$}, that
is of $R \bmod x^k$ for some given parameter $k$.  This kind of
question appears for instance in the algorithms
of~\cite{GiLeSa01,LeSc12}, where we want two terms in the expansion,
so that $k=2$. A related example, in a slightly more involved setting,
involves the evaluation of the second derivative of some
subresultants, for input polynomials in $\K[x,y,z]$~\cite{ImMoPo15}.

Of course, one could simply compute $R$ itself and truncate it
afterwards; however, it seems wasteful to compute all $d^2$ terms of
$R$, incurring a cost of $\O\tilde{~}(d^3)$, before discarding many of
them. Now, for all but finitely many values $a$ in $\K$, it is
possible to compute $R \bmod (x-a)^k$ using $\O\tilde{~}(d k)$
operations in $\K$: indeed, as soon as the non-zero subresultants of
$P$ and $Q$ do not vanish at $a$, we can run the Knuth-Sch\"onhage
algorithm with coefficients truncated modulo $(x-a)^k$, without attempting
to invert an element that would vanish at $a$. The running time
claimed above then follows from the fact that arithmetic operations in
$\K[x]/(x-a)^k$ can be done using $\O\tilde{~}(k)$ operations in $\K$
(for such standard complexity results, our reference
is~\cite{GGbook13}).

If however we cannot choose the expansion point, as is the case here,
there is no guarantee that all divisions remain feasible in
$\K[x]/\langle x^k\rangle$; attempting to divide by an element of positive valuation
would entail a loss of $x$-adic precision. 

An obvious solution is to use a {\em division-free} algorithm over
$\K[x]/\langle x^k\rangle$ (by contrast, so-called fraction-free algorithms often
require the base ring to be a domain).  The best result we are aware
of is due to Kaltofen and Villard~\cite{KaVi04}, with a cost of
$\O(d^{2.698})$ ring operations to compute the determinant of a matrix
of size $d$ over any ring, and thus $\O\tilde{~}(d^{2.698} k)$
operations in $\K$ to solve our problem.

In~\cite{Caruso15}, Caruso studies the phenomenon of loss of precision
in  the iterative version of the subresultant algorithm. He  shows that on
average, if the base field is finite, this loss of
precision grows
linearly with the degree of the inputs. In that same reference, he
also shows how to modify this algorithm to reduce the loss of
precision, resulting in a cost of $\O\tilde{~}(d^2 (k + \delta))$ operations in
$\K$, where
$\delta$ is the maximum of the $x$-adic valuations of the non-zero leading
subresultants of the input polynomials (under the assumption that
these leading subresultants all have valuation less than $k/2$). When
$\K$ is finite, the expected value of $\delta$ is $\O(\log(d))$, so
that the average running time becomes $\O\tilde{~}(d^2 k)$.

Our main result is a complexity estimate for the computation of $R
\bmod x^k$, using the Knuth-Sch\"onhage divide-and-conquer
algorithm. We show that we can compute $R \bmod x^k$ using
$\O\tilde{~}(d k)$ base field operations, when $\K$ has characteristic 
zero, or at least $k$.

We proceed in three steps. First, we compute cofactors $U,V$ and an
integer $t$ such that $UV + PQ = x^t \bmod x^{t+1}$ holds in
$\R=\K[x,y]/\langle x^k\rangle$; this is done in Section~\ref{sec:pseudoinv} by a
suitable adaptation of the half-gcd algorithm. From this equality, we
will able to deduce a first-order linear differential equation
satisfied by the resultant $R$ (Section~\ref{sec:diffeq}). Solving
this differential equation is straightforward, once an initial
condition is known; Section~\ref{sec:first-coeff} shows how to compute
the first non-zero term in the resultant.


\section{Preliminaries}

To a polynomial $P$ of $\R$, we associate a valuation $\v(P)$
defined as the smallest exponent on $x$ in the monomials of $P$ for
$P$ non-zero; by convention, $\v(0) = k$. By Gauss' Lemma, for
any polynomials $P$ and $Q$ in $\R$, we have $\v(PQ) = \v(P) + \v(Q)$
if the sum is less than $k$ (otherwise, $PQ=0$). In all the algorithms
of the paper, a polynomial $P$ of $\R$ is represented by the monomial
$x^{\v(P)}$ and the polynomial $P/x^{\v(P)}$ of valuation 0. We do not
take into account the boolean cost of exponent manipulations, in the
sense that we assume that we can add two valuations in constant time.
We define the valuation $\v(P,Q)$ of a pair of polynomials
$(P,Q) \in \R^2$ as the minimum of the valuations of $P$ and
$Q$. We let $\Seq_k \subset \R^2$ be the set of all pairs of polynomials
$(P,Q)$ with $\v(P)<\v(Q)$, to which we adjoin $(0,0)$. 

For $n\in\mathbb N$, we define $\Pi_n$ as the function from
$\R$ to itself, and by extension from $\Seq_k$ to itself, such
that:
\begin{align*}
    \Pi_n(P) &= P \remx x^{\v(P)+n}\\
    \Pi_n(P,Q) &= (P \remx x^{\v(P,Q)+n}, Q \remx x^{\v(P,Q)+n});
\end{align*}
if $\v(P)+n \ge k$ in this definition, we simply replace it by $k$. We
denote by $\U(\R)$ the set of invertible elements of $\R$; they are
exactly the polynomials $P$ such that $\v(P)=0$ and $\Pi_1(P)\in\K$.

We call {\em degree} of an element of $\K[x]/\langle x^k\rangle$ the degree of its
canonical lift to $\K[x]$.  We will often write expressions of the
form $Q=P/x^{\v(P)}$, for some $P$ in $\R$. Such a quotient is not
unique; we remove ambiguities by requiring that all coefficients of
$Q$ have degree in $x$ less than $k-\v(P)$.

If $P$ is a polynomial in $\R$, then $\Xi P \in \K[y]$ is the
coefficient of $x^i$ in it. Then, we say that $P$ is \emph{normal} if
$\deg_y(\Xv{P}) > \deg_y(\Xi{P})$ for all $i > \v(P)$. Equivalently,
this is the case if the leading coefficient of $P$
has the form $c x^{\v(P)}$, for some non-zero $c$ in $\K$.

Being normal is a useful property. For instance, it allows us to
define quotient and remainder in a Euclidean division in many cases:
if $P \in \R$ is normal and $Q \in \R$ has valuation at least $\v(P)$,
there exist $U$ and $R$ in $\R$ with $\deg(R) < \deg(P)$, and such
that $Q = UP + R$; $U$ is uniquely defined modulo $x^{k-\v(P)}$,
whereas $R$ is unique. We will write $Q \mydivy P$ for the unique
quotient $U$ with degree in $x$ less than ${k-\v(P)}$,
and $R=Q \remy P$.

The following normalization property will be essential for our
algorithms. Below, $\sfM(n)$ is a bound on the number of arithmetic
operations in $\K$ for computing the product of two polynomials $f$,
$g$ in $\K[x]$ of degree at most $n$; we assume that $\sfM$ satisfies
the super-linearity conditions of~\cite[Chapter~8]{GGbook13}. Using
the Cantor-Kaltofen algorithm~\cite{CaKa91}, we can take $\sfM(n)$ in
$\O(n \log(n) \log\log(n))$.  Using Kronecker's
substitution as in~\cite{GaSh92}, multiplication or
Euclidean division in degree $d$ in $\R$ can then be done in
$\O(\sfM(kd))$ operations in~$\K$.

\begin{Lemma}[Normalization]
    \label{lem:normalization}
    For $P$ non-zero in $\R$, there exist polynomials $U$ and $T$ in
    $\R$ such that $T$ is normal, $U$ is a unit, with $\Pi_1(U)=1$,
    and $P = UT$.

    The polynomial $T$ is unique, whereas $U$ is uniquely defined
    modulo $x^{k-\v(P)}$.  Moreover, if $P$ has degree $d$, for $n \le
    k-\v(P)$, given $\Pi_n(P)$, $\Pi_n(U)$ and $\Pi_n(T)$ can be
    computed using $\O(\sfM(dn))$ operations in $\K$.
\end{Lemma}
\begin{proof}
  Up to dividing $P$ by $x^{\v(P)}$, we may assume that $P$ has
  valuation $0$. Then, uniqueness and the corresponding algorithm are
  from Algorithm~Q in~\cite{Musser1975}; the cost analysis is a
  straightforward extension of that given for Hensel step
  in~\cite[Chapter~15]{GGbook13} using Kronecker substitution for the
  arithmetic operations. Once the result is known for
  $P/x^{\v(P)}$, we recover our claim by multiplying $T$ by
  $x^{\v(P)}$.
\end{proof}

Given $P$ in $\R$, we denote by respectively $\Lc(P)$ and $N(P)$ the
unique $U$ and $T$ that satisfy the above conditions, with $U$ having
degree in $x$ less than $k-\v(P)$; these two polynomials have degree
in $y$ at most $d$, and $\v(N(P))=\v(P)$. We also define
$\Lc_n(P)=\Pi_n(\Lc(P))$ and $N_n(P)=\Pi_n(N(P))$; by
the previous lemma, if $P$ has degree $d$, for any $n$, $\Lc_n(P)$ and
$N_n(P)$ can be computed in time $\O(\sfM(dn))$.

Using this result, we can reduce the problem of computing the
resultant of two polynomials in $\R$ to the problem of computing the
resultant of two normal polynomials.

\begin{Lemma}
  \label{lem:normalres}
  Let $P$ and $Q$ be in $\R$ of degrees in $y$ at most
  $d$. Then there exist four monic polynomials $A,B,C,D$ in $\R$, with
  degrees at most $d$, and $u$ in $\K[x]/\langle x^k\rangle$ such that
  $\Res(P,Q) = u\Res(A,B)\Res(C,D)$.
  Moreover, $A,B,C,D,u$ can be computed in time $\O(\sfM(dk))$.
\end{Lemma}
\begin{proof}
By the multiplicative property of the resultant, 
  $\Res(P,Q) = \Res(\Lc(P), Q) \Res(N(P), Q)$. Let $d_c,d_n,d_Q$ be the
  degrees in $y$ of $\Lc(P), N(P)$ and $Q$ respectively.

  Taking the reciprocal polynomial of $\Lc(P)$ and $Q$ changes at
  most the sign of their resultant. Let us thus define $\tilde P =
  y^{d_c}\Lc(P)(1/y)$ and $\tilde Q = y^{d_Q}Q(1/y)$ and let
  $c_0\in\K[x]/\langle x^k\rangle$ be the leading coefficient of $\tilde P$.  By
  construction, $c_0$ is a unit, so we can define $A=\tilde P/c_0 \in
  \R$, and we have $\Res(\Lc(P),Q)=(-1)^{d_c d_Q} {c_0}^{d_Q}\Res(A,
  \tilde Q)$, where $A$ is  monic in $y$.  Let $R$ be the
  remainder of the Euclidean division of $\tilde Q$ by $A$ and let $B
  = R + A$; since $A$ is monic, $\Res(A, \tilde
  Q) = \Res(A,B)$.

  Similarly, because $N(P)$ is normal, we can write it as $N(P) =
  n_0 C$, where $n_0 \in \K[x]/\langle x^k\rangle$ is its leading coefficient and $C$
  is monic in $y$; then, we have $ \Res(N(P), Q)= {n_0}^{d_Q}
  \Res(C,Q)$. Defining as above $D = (Q\remy
  C) + C$, we deduce $\Res(N(P),Q) = {n_0}^{d_Q}\Res(C, D)$, and finally
  $$\Res(P,Q) = (-1)^{d_c d_Q} (c_0n_0)^{d_Q}\Res(A,B)\Res(C,D).$$
  $A$ is monic of degree at most $d$, since $\Lc(P)$,
  has degree at most $d$. The remainder $R$
  has degree less than the degree of $A$, so that $B=R+A$ is monic of the same
  degree as $A$.  The same holds for $C$ and $D$.

  In terms of complexity, after computing $\Lc(P)$ and $N(P)$, all
  other operations are $\O(d)$ inversions or multiplications of power
  series in $\K[x]/\langle x^k\rangle$, and $\O(1)$ Euclidean divisions by monic
  polynomials of degree at most $d$ in $\R$. Their total cost is 
 $\O(\sfM(dk))$ operations in $\K$.
\end{proof}


\section{Computing pseudo inverses}\label{sec:pseudoinv}

In this section, we show that given two polynomials $P,Q$ in $\R$, we
can compute a matrix $\mat{M}=\left (\begin{smallmatrix}U & V \\ X &Y
\end{smallmatrix} \right )\in \scrM_2(\R)$ such that $UP+VQ$ is of the form $x^t \bmod x^{t+1}$,
for some integer $t$; we call $U$ and $V$ {\em pseudo-inverses} of $P$
and $Q$.

To simplify notation, for $\s=(P,Q)$ and $\mat{M}$ as above, 
we simply write $\mat{M} \cdot s$ for the matrix-vector product
$\mat{M} \left (\begin{smallmatrix} P \\ Q 
\end{smallmatrix}\right )$.


\subsection{The pseudo-division operator $\Q$}\label{ssec:Q}

In this
subsection, we define an operator $\Q: \Seq_k \to \Seq_k$ and study
its properties. For $\s=(0,0)$, we define $\Q(\s)=\Id$; otherwise, the
construction involves three stages.

For a non-zero pair of polynomials $\s=(P,Q)$ in $\Seq_k$, define
$$\eta(s) := \v(Q \remy N(P)) - \v(P).$$ This is well-defined, as
 $N(P)$ has the same valuation as $P$. 
\begin{Lemma}\label{lem:computeEta}
  For an integer $n$ and $\s$ in $\Seq_k$, given $\Pi_n(\s)$, we can
  compute $\min(\eta(s),n)$ using $\O(\sfM(dn))$ operations in $\K$.
\end{Lemma}
\begin{proof}
  We start by dividing both $P$ and $Q$ by $x^{\v(P)}$; this does not
  involve any arithmetic operation. We can then compute the
  normalization $N_n(P)/x^{\v(P)}$, and do the Euclidean division of
  $Q/x^{\v(P)}$ by this polynomial, using coefficients taken modulo
  $x^n$, both in time $\O(\sfM(dn))$. The valuation of the remainder
  is precisely $\min(\eta(s),n)$.
\end{proof}
The next lemma shows a more intrinsic characterization of~$\eta$. We
denote by $\sigma_0 : \R \rightarrow \K[y]$ the evaluation morphism that
sends $x$ to $0$.

\vspace{-2mm}

\begin{Lemma}\label{lem:Euclidean0}
For any integer $t \ge 0$, with $J_t(\s)=\left<P, Q\right>:x^{\v(\s)+t}$ and
$I_t(\s) = \sigma_0\left(J_t(\s)\right)$, we have
$$\begin{cases}
  I_t(\s) = I_0(\s)\text{ if } 0 \leq t < \eta(\s)\\
  I_t(\s) \varsupsetneq I_0(\s)\text{ if } t\geq \eta(\s).
\end{cases}$$
\end{Lemma}
\begin{proof}
For $t\geq\v(\s)$, we have $I_0(\s)\subset I_t(\s)$. Then, let $Q' =
Q\remy N(P)$. If
$0\leq t<\eta(\s)$ let $W$ be such that $x^{\v(s)+t}W= UP+VQ' \in
\langle P,Q\rangle$. Since $x^{\v(\s)+\eta(\s)}$ divides $Q'$
and $t+1\leq \eta(\s)$ we have $x^{\v(s)+t}W = UP \mod x^{\v(s)+t+1}$,
so $x^{t}$ divides $U$ and $x^{\v(s)}W = U'P \mod x^{\v(s)+1}$.
Thus $I_t(\s) \subset I_0(\s)$ and $I_t(\s) = I_0(\s)$.
If $t\ge\eta(\s)$, $I_t(\s)$ contains the residue class of the remainder
of $\Pi_1(Q)/x^{\v(Q)}$ by $\Pi_1(P)/x^{\v(P)}$, that is
a non-zero polynomial of degree less than $\deg_y(\Pi_1(P))$, and is
not included in $I_0(\s)$.
\end{proof}

For $\s=(P,Q)$ as above, perform a Euclidean division of
$Q$ by the normal polynomial $N_{\eta(\s)}(P)$, and define the matrix
$$
\mathbf D_{\s} := \begin{pmatrix}
                    x^{\eta(\s)} & 0\\
         - \left(Q \mydivy N_{\eta(\s)}(P)\right) \remx x^{\eta(\s)}
                    & \Lc_{\eta(s)}(P)
                   \end{pmatrix}.$$
Then, the polynomial $\tilde Q$ defined by
$$ \mathbf D_{\s} \cdot \begin{pmatrix} P\\ Q\end{pmatrix}
  = \begin{pmatrix} x^{\eta(s)} P\\ \tilde Q \end{pmatrix}
$$ has valuation $\v(P)+\eta(s)$ and we have the inequality \sloppy
  $\deg_y(\Pi_1(\tilde Q))<\deg_y(\Pi_1(P))$.

Given $P,Q\in \K[y]$ of degree at most $d$, and denoting by $G$ their
monic gcd, there exists an invertible matrix $\mathbf G_{(P,Q)}$ with
entries in $\K[y]$ of degree less than $d$ such that
$$
\mathbf G_{(P,Q)}\cdot\begin{pmatrix} P\\ Q \end{pmatrix}
=
\begin{pmatrix} G \\ 0 \end{pmatrix}.
$$
If $Q=0$ then $G_{(P,Q)}$ is the identity matrix.
By extension, for $s=(P,Q)\in \R^2$ with $\v(P)=\v(Q)$, we define
$\mathbf G_{s}$ as $\mathbf G_{(\Pi_1(P)/x^{\v(P)},\Pi_1(Q)/x^{\v(Q)})}$. By construction,
the entries $(\tilde P,\tilde Q)$ of $\mathbf G_{s}\cdot s$ 
satisfy $\v(\tilde P) < \v(\tilde Q)$, or $\tilde P = \tilde Q = 0$; in other words,
the pair $\mathbf G_{s}\cdot s$ belongs to $\Seq_k$.

For a pair $s=(P,Q)$ in $\Seq_k$ with $P \neq 0$ we define
$\mathbf N_s$ by
$$\mathbf N_s = \begin{pmatrix}
  \Lc(P)^{-1} & 0 \\
  - \Lc(P)^{-1}(Q \mydivy N(P)) & 1
\end{pmatrix}.$$
If $P=0$, set $\mathbf N_s=\Id$. Else, $v(Q)>v(P)$ and
there exists a matrix $\mathbf R$ such that $\mathbf N_s =
\Id+x\mathbf R$; then, $\mathbf N_s$ is invertible.
\begin{Lemma}
    \label{lem:Euclidean}
    For $\s = (P,Q) \in \Seq_k^2$, define
    $$\Q(\s) := \mathbf N_{\mathbf G_{\mathbf D_\s\cdot\s}
                        \cdot\mathbf D_\s \cdot \s}
             \cdot \mathbf G_{\mathbf D_\s  \cdot \s } \cdot \mathbf D_\s.$$
    Then, 
    \begin{enumerate}
    \item $\Q(\s) \cdot \s$ is in $\Seq_k$;
    \item the following equality between ideals holds:
      $$\left< \Q(\s) \cdot \s \right> = \left< \s \right> \cap \left<
      x^{\v(\s)+\eta(\s)} \right>;$$
    \item  $\Pi_1(\Q(\s)\cdot \s) = \left(\begin{smallmatrix}
     x^{v(\s)+\eta(\s)} G\\0\end{smallmatrix}\right)$, where $G$ is the gcd of $\Pi_1(P)/x^{v(\s)}$ and $\Pi_1(\tilde Q)/x^{v(\s)+\eta(\s)}$;
    \item $\Q(\s)\cdot \s$ and $\Q(\s')\cdot\s'$ generate the same
      ideal, for any pair $\s'$ that generates the same ideal as $\s$;
    \item the entries of $\Q(s)\cdot \s$ have degree less than
      $\deg_y(\Pi_1(P))$.
    \end{enumerate}
    For $\s=(P,Q) \in\Seq_k$, where $P$ and $Q$ have degree
    at most $d$, given $\Pi_{n}(\s)$ for some $n \ge \eta(\s)$,
    $\Pi_{n-\eta(\s)}(\Q(\s)\cdot\s)$ can be computed using
    $\O(\sfM(dn)+\sfM(d)\log(d))$ operations in $\K$.
\end{Lemma}
\begin{proof}
  We saw just above the lemma that $\mathbf G_{\mathbf D_\s\cdot\s}
  \cdot (\mathbf D_\s \cdot \s)$ belongs to $\Seq_k$. Since applying
  the matrix $\mathbf N_{\mathbf G_{\mathbf D_\s\cdot\s} \cdot \mathbf
    D_\s \cdot \s }$ to this vector does not change the valuations of
  its entries, we deduce that $\Q(\s) \cdot \s$ is in $\Seq_k$.

  In order to prove the relation $\langle \Q(\s) \cdot \s
  \rangle = \langle \s \rangle \cap \langle x^{\v(\s)+\eta(\s)} \rangle$, it
  is sufficient to prove that $\langle \mathbf D_\s \cdot \s\rangle =
  \langle \s \rangle \cap \langle x^{\v(\s)+\eta(\s)} \rangle$, since the
  matrices $\mathbf N_{\mathbf G_{\mathbf D_\s\cdot\s} \cdot \mathbf
    D_\s \cdot \s }$ and $\mathbf G_{\mathbf D_\s\cdot \s}$ are
  invertible over $\R$. 
  The two polynomials in the vector $\mathbf
  D_\s \cdot \s=(x^{\eta(\s)}P, \tilde Q)$ are divisible by
  $x^{\v(P)+\eta(\s)}$, thus $\langle \mathbf D_\s \cdot \s\rangle \subset
  \langle \s \rangle \cap \langle x^{\v(P)+\eta(\s)} \rangle$.
  For the other inclusion, let $W$ be a polynomial in $ \langle \s
  \rangle \cap \langle x^{\v(P)+\eta(\s)} \rangle$. Since we have $s =
  \langle P,Q\rangle = \langle P, \tilde Q\rangle$, we can write $W= U
  P + V \tilde Q$, for some $U,V$ in $\R$. On the other hand, we know
  that $x^{\v(P)+\eta(\s)}$ divides $W$, and since it also divides
  $\tilde Q$, it divides $UP$. This implies that $x^{\eta(\s)}$
  divides $U$, which means that $W$ is in $\langle
  \mathbf D_\s \cdot \s \rangle$. This allows us to conclude for the
  equality of ideals, noting that $\v(\s)=\v(P)$.

  The third point comes from the fact that the matrix
  $\mathbf N_{\mathbf G_{\mathbf D_\s\cdot\s} \cdot\mathbf D_\s \cdot
  \s}$ can be written $\Id + x\mathbf R$, so that
  $\Pi_1(\Q(\s)\cdot\s)
  = \Pi_1(\mathbf G_{\mathbf D_\s  \cdot \s } \cdot \mathbf D_\s \cdot
  \s)
  = \left(\begin{smallmatrix} x^{v(\s)+\eta(\s)} G\\0\end{smallmatrix}\right)$.
    
    If $\s$ and $\s'$ generate the same ideal,
    $\eta(\s)=\eta(\s')$ (Lemma~\ref{lem:Euclidean0}); using the second 
    item, this proves point 4.


    The degree property follows from the fact that in the pair
    $(A,B)=\Q(\s)\cdot\s$, $A$ is normal of degree $\deg_y(G) <
    \deg_y(\Pi_1(P))$, and $B$ is a remainder modulo $N(A)$.

Computing $\mathbf
  D_\s$ can be done in time $\O(\sfM(d\eta(\s)))$ using
  Lemma~\ref{lem:normalization} to compute $N_{\eta(\s)}(P)$ and
  $\Lc_{\eta(s)}(P)$ in time $\O(\sfM(d\eta(\s)))$. The matrix-vector
  product that gives $\tilde Q$ is done in degree $n$ in $x$ and $d$
  in $y$, in time $\O(\sfM(d n))$. Then, we saw that $\Pi_1(\tilde Q)$
  has degree less than $d$, so  computing $\mathbf G_{\mathbf D_\s
    \cdot\s}$ is an extended gcd calculation in $\K[y]$ that can be
  done in time $\O(\sfM(d)\log( d)$.
  Applying $\mathbf G_{\mathbf D_\s \cdot\s}$ to ${\mathbf D_\s
  \cdot\s}$ takes $\O(n\sfM(d))$, and results in a matrix
  of degree $\O(d)$ in $y$. Finally, applying $\mathbf N_{\mathbf
    G_{\mathbf D_\s \cdot\s} \cdot \mathbf D_\s \cdot\s}$ to $\mathbf
  G_{\mathbf D_\s \cdot\s} \cdot \mathbf D_\s \cdot\s$ is again a
  normalization at precision $n$ along with arithmetic operations that
  can all be done in time $\O(\sfM(dn))$.
\end{proof}


\subsection{An extension of the half-gcd}

Our goal is now to iterate the pseudo-division $\Q(s)$ until we reach
an integer $t$ such that $\langle s\rangle\cap\langle x^t\rangle =
\langle x^t\rangle$. We use a divide-and-conquer algorithm, inspired
by the half-gcd algorithm. If we applied this idea directly, the increase in
degree in $y$ of the transition matrices would prevent us from getting
a softly linear bound in the degree of the input; we will thus work
modulo an equivalence relation, to control the size of the
intermediate polynomials.

Consider the following equivalence relation on $\Seq_k$: for any two
pairs $(P,Q)$ and $(P',Q')$ in $\Seq_k$, we say that $(P,Q) \sim
(P',Q')$ if and only if the ideals they generate are the same.  In
particular, this implies that $\v(P,Q)=\v(P',Q')$ and that $\Pi_1(P)$
and $\Pi_1(P')$ are equal up to a constant.

Let further $\H, \H'$ be two functions from $\Seq_k$ to
$\scrM_2(\R)$. Extending the equivalence property to the set of
functions, we say that $\H$ and $\H'$ are equivalent if for all $\s\in
\Seq_k$ we have $\H(\s) \cdot \s \sim \H'(\s) \cdot \s$. We still
write in this case $\H \sim \H'$.

\begin{Definition}[Euclidean function]
We say that $\H:\Seq_k\rightarrow\scrM_2(\R)$ is \emph{Euclidean} if:
\begin{itemize}
  \item for all $\s,\s' \in \Seq_k$, if $\s \sim \s'$, then $\H(\s)
    \cdot \s \sim \H(\s') \cdot \s'$
  \item for all $\s \in \Seq_k$, if $\s$ is non-zero, $\v(\H(\s) \cdot \s) > \v(\s)$.
\end{itemize}
We denote $\v(\H(\s).\s)-\v(\s)$ by $\eta_\H(\s)$,
and we say that
$\H$ is \emph{online} if for all $\s \in \Seq_k$:
\begin{itemize}
  \item $\eta_{\H}(\Pi_i(\s)) \geq i$ for all $i \leq \eta_{\H}(\s)$
  \item $\H(\s) \cdot \s
    \sim \H(\Pi_{i}(\s)) \cdot \s$ for all $i \geq \eta_{\H}(\s)+1$.
\end{itemize}
\end{Definition}

\begin{Lemma}\label{lem:QisEuclidean}
  The function $\Q$ introduced in Lemma~\ref{lem:Euclidean} is an
  online Euclidean function.
\end{Lemma}
\begin{proof}
  The first point was proved in Lemma~\ref{lem:Euclidean}.
  Now, let $\s = (P,Q)$ be a non-zero element of $\Seq_k$. By
  construction, we have $\v(\Q(\s)\cdot\s) = \v(\s) + \eta(\s)$. In
  particular, $\eta(\s) = \v(Q \remy N(P)) - \v(P) \geq \v(Q) - \v(P)
  > 0$. Thus, $\Q$ is a Euclidean function.

  To prove that $\Q$ is online,  notice that $\Q(\s \remx
  x^{\v(\s) + i}) \remx x^i=$ $ \Q(\s) \remx x^i$. In particular for
  $i\leq\eta_\Q(\s)$,  $\Q(\Pi_i(\s)) \cdot \Pi_i(\s)$ is 
given by
  \begin{align*}
    &  \left(\Q(\s)\remx x^i\right) \cdot
      \left(\s \remx x^{\v(\s)+i}\right) \bmod x^{\v(s)+i}\\
    & = \Q(\s) \cdot \s \mod x^{\v(\s)+i} = (0,0) \mod x^{\v(\s)+i},
  \end{align*}
  so that $\eta_{\Q}(\Pi_i(\s))\geq i$.

  If $i>\eta(\s)$, we  prove that there exists a matrix $\mathbf R$
  such that $\Q(\Pi_i(\s)) = (\Id+x\mathbf R) \cdot \Q(\s)$. Indeed, if
  that holds,  $(\Id+x\mathbf R)$ is an invertible
  matrix over $\R$, so that $\left<\Q(\Pi_i(\s))\cdot\s\right> =
  \left<\Q(\s)\cdot\s\right>$.  Moreover $\Pi_1(\Q(\Pi_i(\s))\cdot\s)
  = \Pi_1(\Q(\s)\cdot\s)$, which will be sufficient to conclude. Let
  $\mathbf M = \mathbf G_{\mathbf D_{\s}\cdot\s}\cdot \mathbf D_{\s}$.
  If $i>\eta(\s)$, we note by construction that $\mathbf M \remx x^i =
  \mathbf M$.  Moreover there exists a matrix $\mathbf R$ such that
  $\mathbf N_{\mathbf M\cdot\s} \remx x^i = (\Id+x\mathbf R)\cdot
  \mathbf N_{\mathbf M\cdot\s}$, so that $\Q(\Pi_i(\s)) =
  (\Id+x\mathbf R) \cdot \Q(\s)$.
\end{proof}
Let $\s$ be an element of $\Seq_k$ and define by recurrence the
following sequence of elements of $\scrM_2(\R)$:
\begin{equation}
    \label{eq:recursivesequence}
\begin{cases}
    \mathbf Q_0     &=\quad \Id\\
    \mathbf Q_{n+1} &=\quad \Q(\mathbf Q_n \cdot \s) \cdot \mathbf Q_n.
\end{cases}
\end{equation}
For $\s$ in $\Seq_k$, the sequence $(\v(\mathbf Q_i \cdot \s))_{i\in\mathbb N}$ is
increasing, until it reaches $\v(\mathbf Q_i \cdot \s)=k$. Thus
given an integer $n$, we can define the function $\Q_n$ from $\Seq_k$ to
$\scrM_2(\R)$ by
\begin{align*}
\Q_n: \Seq_k & \rightarrow \scrM_2(\R)\\
       \s   & \mapsto \mathbf Q_{i_0},\ i_0=\max\{i\in\mathbb N \mid \v(\mathbf Q_i \cdot \s)-\v(\s) \leq n\}
\end{align*}
for $\s$ non-zero; for $\s=(0,0)$, we set $\Q_n(\s)=\Id$.  In
particular, for any $\s$, $\Q_0(\s)=\Id$ (since for $\s$ non-zero,
$\mathbf Q_1 \cdot \s =\Q(\s)\cdot s$ has valuation greater than that
of $\s$).

\begin{Lemma}
  \label{lem:Qn}
  For $n \ge 0$, if $\s\sim\s'$ in $\Seq_k$ then $\Q_n(\s)\cdot\s \sim
  \Q_n(\s')\cdot\s'$. Moreover, let $j$ be the minimal integer such that
  $I_j(\s) = I_n(\s)$; then, $\langle\Q_n(\s)\cdot \s\rangle =
  \langle\s\rangle \cap \langle x^{\v(\s)+j}\rangle$.
\end{Lemma}
\begin{proof}
  We use a recurrence on $n$. For $n=0$, it is clear that
  $\Q_0$ satisfies the desired properties. Then, given $n \geq 1$ and
  $\s \sim \s'$ two equivalent elements of $\Seq_k$ we know by
  recurrence assumption that $t := \Q_{n-1}(\s)\cdot\s \sim t' :=
  \Q_{n-1}(\s')\cdot\s'$.

  If $\v(\Q(t) \cdot t) > \v(s) + n$, then $\v(\Q(t')\cdot t') > \v(s') +
  n$ and $\Q_n(\s) = \Q_{n-1}(\s)$ and $\Q_n(\s') = \Q_{n-1}(\s')$,
so that $\Q_{n}(\s)\cdot\s \sim \Q_{n}(\s')\cdot\s'$. Moreover, in
  this case, let $j$ be the minimal integer such that
  $I_j(\s) = I_{n-1}(\s)$. For all $\ell < \eta(t)$ we have $I_\ell(t)
  = I_0(t)$. In particular, $\ell := n - (\v(t)-\v(\s)) < \eta(t)$,
  and $I_\ell(t) = I_n(\s)$ and $I_0(t) = I_{n-1}(\s)$. Hence,
  $\langle\Q_n(\s)\cdot\s\rangle = \langle\s\rangle \cap
  \langle x^{\v(\s)+j}\rangle$ and $j$ is the smallest integer such that $I_n(\s)
    = I_{n-1}(\s) = I_j(\s)$. 

  Otherwise, $\v(\s) + n-1 < \v(\Q(t) \cdot t) \leq \v(s) + n$, so
  that $\v(\Q(t) \cdot t) = \v(\s)+n$. Then $\v(\Q(t')\cdot
  t') = \v(s') + n$. Thus $\Q_n(\s) = \Q(t) \cdot \Q_{n-1}(\s)$
  and $\Q_n(\s') = \Q(t') \cdot \Q_{n-1}(\s')$. Since $\Q$ is a
  Euclidean function, $\Q(t)\cdot t \sim \Q(t')\cdot t'$ and this leads
  to $\Q_n(\s)\cdot\s \sim \Q_n(\s')\cdot \s'$. Moreover, let $\ell$ be
  the smallest integer such that $I_\ell(t) \neq I_0(t)$. We have
  $\langle\Q(t)\cdot t\rangle = \langle t\rangle \cap \langle
  x^{\v(t)+\ell}\rangle$. According to Lemma~\ref{lem:Euclidean0} we have
  $\ell = \eta(t) = \v(\s) + n - \v(t)$. In particular $n = \v(t) + \ell
  -\v(\s)$, and we have $I_n(\s) = I_\ell(t)$ and for all $i<n$ we have
  $I_i(\s) \neq I_n(\s)$. Thus, $n$ is the smallest integer $j$ such
  that $I_j(\s) = I_n(\s)$ and $\langle\Q_n(\s)\cdot\s)\rangle =
  \langle\Q(t)\cdot t\rangle = \langle t \rangle \cap \langle x^{\v(t) +
  \eta(t)} \rangle$. Since $\langle t\rangle =
  \langle\Q_{n-1}(\s)\cdot\s\rangle = \langle s\rangle \cap
  \langle x^{\v(\s)+i}\rangle$ for some $i\leq n-1$, this leads to
  $\langle\Q_n(\s)\cdot\s)\rangle = \langle \s \rangle \cap \langle x^{\v(\s)
  + n} \rangle$.
\end{proof}

Our goal is to compute efficiently a mapping equivalent to~$\Q_n$. To
this effect, we introduce in Algorithm~\ref{algo:generichalfgcd} a
generalized version of the half-gcd algorithm~\cite[p. 320]{GGbook13}.
In order to control the degree in $y$ of the intermediate elements in
Algorithm~\ref{algo:generichalfgcd}, we
use a function $\varphi_{\s,n}:\scrM_2(\R)\to\scrM_2(\R)$,
that depends on an element $\s\in\Seq_k$ and an integer $n$.  

Let
$\s=(P,Q)$ be a pair of polynomials in $\Seq_k$, $n$ be an integer,
and let $\mathbf M = \left(\begin{smallmatrix} A & B\\ C & D
\end{smallmatrix}\right)$
be a matrix of $\scrM_2(\R)$. If $Q = 0$ or $\v(\mathbf M\cdot\s) =
\v(\s)$ then we let
$\varphi_{\s,n}(\mathbf M) = \Id$.
Otherwise, neither $P$ nor $Q$ are zero and we define
$\varphi_{\s,n}(\mathbf M)$ as the remainder of
\begin{align*}
  &\begin{pmatrix}
    L_Q(L_PA\remy N_u(Q))~~L_P(L_QB\remy N_u(P))\\
    L_Q(L_PC\remy N_u(Q))~~L_P(L_QD\remy N_u(P))
  \end{pmatrix}
\end{align*}
by $x^{n+1}$,
where $L_P,L_Q$ are $\Lc(P), \Lc(Q)$ and $N_u(P), N_u(Q)$ are two
polynomials with constant leading coefficients
such that $N(P) = x^{\v(P)} N_u(P)$ and $N(Q) = x^{\v(Q)} N_u(Q)$.

\begin{Lemma}[Size control]\label{lemma:varphi}
  Let $\s \in \Seq_k$ and assume $\mathbf M \cdot \s$ is also in
  $\Seq_k$.
  If $\mathbf M\cdot \s \sim \Q_n(\s)\cdot\s$, let $j=\v(\mathbf
  M\cdot\s)-\v(\s)$. Then $\varphi_{\s,j}(\m)\cdot\s\in\Seq_k$ and $\varphi_{\s,j}(\m)\cdot\s \sim \m \cdot
  \s$.
  
  Moreover the degree in $y$ of the entries in
  $\varphi_{\s,j}(\m)$ are bounded by $d_Q -1$ for the first column
  and $d_P-1$ for the second one, with $d_P=\deg_y(P)$ and $d_Q=\deg_y(Q)$.
\end{Lemma}
\begin{proof}
  The degree bound follows from $\deg(L_P)+\deg(N_u(P)) = d_P$, and
similarly for $Q$.
If $Q=0$ or $\v(\mathbf M\cdot\s) = \v(\s)$ then $\Q_n(\s) = \Id$ for all $0 \leq n < k$, and
  $\varphi_{s,j}(\mathbf M) = \Id = \Q_n(\s)$.
  
  Now let $s=(P,Q)$ a pair in $\Seq_k$ with $Q\neq 0$ and $\v(\mathbf
  M\cdot\s)>\v(\s)$ and let
  $\left(\begin{smallmatrix}G\\H\end{smallmatrix}\right) = \mathbf{M}
    \cdot \left(\begin{smallmatrix}P\\Q\end{smallmatrix}\right)$.  By
      assumption, $\langle G,H\rangle  = \langle P,Q\rangle \cap \langle
      x^{\v(\s)+j}\rangle$. Moreover, $\Pi_1(G,H) = (\Pi_1(G), 0)$ and
      Lemma~\ref{lem:Euclidean} shows that $\deg_y(\Pi_1(G))< d_P$.
      
      Assume first that $P$ and $Q$ are
      normal. In this case, $L_P = L_Q = 1$ and we have
      $$\varphi_{\s,j}(\mathbf M) \cdot
      \left(\begin{matrix}P\\Q\end{matrix}\right) = 
        \left(\begin{matrix}
            G+KN_u(P)N_u(Q)+Lx^{j+1}\\
            H+MN_u(P)N_u(Q)+Nx^{j+1}
        \end{matrix}\right),$$
        where $K, M$ are polynomials of $\R$ and $L,N$ are in
        $\langle P,Q\rangle $, so that $\v(L)\geq\v(s)$ and $\v(N)\geq\v(\s)$.
        In this case we also know  that the
        degrees of $G+KN_u(P)N_u(Q)+Lx^{j+1}$ and $H+MN_u(P)N_u(Q)+Nx^{j+1}$ are lower
        than $d_P+d_Q$. On the other hand, $G \remx x^{\v(\s)+j+1}$ has a
  degree in $y$ less than $d_P$ and $H \remx x^{\v(\s)+j+1} =
  0$. Also, $Lx^{j+1} \remx x^{\v(\s)+j+1} = Nx^{j+1} \remx
  x^{\v(\s)+j+1} = 0$.
  Hence, $K$ and $M$ have a valuation greater than or
  equal to $\v(\s)+j+1$ and $KN_u(P)N_u(Q)$ and $MN_u(P)N_u(Q)$ are in
  $\langle x^{j+1}P\rangle $.
  
  Thus $\varphi_{\s,j}(\mathbf M) \cdot
  \left(\begin{smallmatrix}P\\Q\end{smallmatrix}\right) 
  = \left(\begin{smallmatrix} G + G'\\ H + H'\end{smallmatrix}\right)$,
  where $G'$ and $H'$ are two polynomials of
  $\langle x^{j+1}P,x^{j+1}Q\rangle $.  In particular, by assumption this
  implies that $G'$ and $H'$ belongs to the ideal $\langle xG,xH\rangle $.
  Thus there is an invertible matrix that sends
$\left(\begin{smallmatrix}G\\H\end{smallmatrix}\right)$ to
$\left(\begin{smallmatrix}G+G'\\H+H'\end{smallmatrix}\right)$ and
  $\langle G,H\rangle  = \langle G+G', H+H'\rangle $ and
  $\Pi_1(G+G',H+H')=\Pi_1(G,H)$. In particular $\varphi_{\s,j}(\mathbf
  M) \cdot \s \in \Seq_k$.

  If $P$ and $Q$ are not normal, let $\mathbf L$ be the
  matrix $\left(\begin{smallmatrix} L_P & 0\\ 0 & L_Q
  \end{smallmatrix}\right)$. In this case, $\m \cdot \s =
  \m\cdot \mathbf L \cdot \left(\begin{smallmatrix} N(P)\\N(Q)
  \end{smallmatrix}\right)$. Hence, using the first part of the proof
  on $N(\s) = (N(P),N(Q))$, we have 
  \begin{align*}
  \varphi_{N(\s),j}(\mathbf M \cdot \mathbf L)\cdot
    \mathbf L^{-1} \cdot \s
  & = \varphi_{N(\s),j}(\mathbf M \cdot \mathbf L) \cdot
    \left(\begin{smallmatrix} N(P)\\N(Q) \end{smallmatrix}\right)\\
  & \sim \mathbf M \cdot \mathbf L \cdot
    \left( \begin{smallmatrix} N(P)\\N(Q) \end{smallmatrix}\right)\sim \mathbf M \cdot \s.
  \end{align*}
  Then, $\varphi_{\s,j}(\mathbf M) =
  \varphi_{N(\s),j}(\mathbf M \cdot \mathbf L)\cdot \left(\det(\mathbf
  L) \mathbf L^{-1}\right) \mod x^{j+1}$. In particular, using the same
  argument as above, this implies
  that
  $\varphi_{\s,j}(\mathbf M)\cdot\s =
  \varphi_{N(\s),j}(\mathbf M \cdot \mathbf L)\cdot \left(\det(\mathbf
  L) \mathbf L^{-1}\right) \cdot \s \mod \left<xG,xH\right>$.
  Finally, we can factor out $\det(\mathbf L)$ and since it is
  invertible, this implies that $\det(\mathbf L) \varphi_{N(\s),j}(\mathbf M \cdot \mathbf L)\cdot
  \mathbf
  L^{-1} \cdot \s$ generates the ideal $\left<G,H\right>$. Finaly, using
  tthe same argument as above, this leads to $\varphi_{\s,j}(\mathbf M)
  \cdot \s \sim \mathbf M \cdot \s$.
\end{proof}

\begin{algorithm}
    \caption{Generalized version of the half-gcd}
    \label{algo:generichalfgcd}
    \begin{algorithmic}
        \Function{GenericHalfGcd}{$n, \s$}
        \State $\tilde \s \gets \Pi_{n+1}(\s)$
        \If{$n = 0 $ or $\tilde \s =(0,0)$}
            \State \Return $\Id$
        \EndIf
        \State $\mathbf R \gets \generichalfgcd (\lfloor\frac n 2\rfloor,
                                    \tilde \s)$
        \State $n' \gets n - (\v(\mathbf R \cdot \tilde \s) -
            \v(\tilde \s))$
        \State $\tilde \u \gets \Pi_{n'+1}(\mathbf R \cdot
            \tilde \s)$
        \State $\eta \gets \min(\eta(\tilde \u), n'+1)$
        \If{$\eta > n'$}
            \State \Return $\mathbf R$
        \EndIf
        \State $\tilde \t \gets \Pi_{n'- \eta+1}
            (\mathbf \Q(\tilde \u) \cdot \tilde \u)$
        \State $\mathbf S \gets \generichalfgcd (
            n - \left(\v(\tilde \t)-\v(\tilde \s)\right),
            \tilde \t)$
            \State $\mathbf M \gets \varphi_{\tilde \s, \v(\mathbf
              S\cdot\tilde \t) - \v(\tilde \s)}
            (\mathbf S \cdot \Q(\tilde \u) \cdot \mathbf R)$
        \State \Return $\mathbf M$
        \EndFunction

    \end{algorithmic}
\end{algorithm}

\vspace{-5mm}

\begin{Lemma}[Generalized half-gcd]
    \label{lem:generichalfgcd}
    For $n \in \mathbb N$, denote by $\H(n,.): s\mapsto \H(n,s)$ the function
    computed by Algorithm~\ref{algo:generichalfgcd}. Then,  $\H(n,.)\sim \Q_n$. 
\end{Lemma}
\begin{proof}
    We prove by recurrence on $n$ that $\H(n,.) \sim \Q_n$.  For
    $n=0$, and for any $\s\in\Seq_k$, we have by construction
    $\H(0,\s) = \Id$. On the other hand, we saw that for all $\s$,
    $\Q_0(\s) = \Id$, so our claim holds.
    Now assume that $\H(i,.)\sim\Q_i$ for $0 \leq i < n$; we prove
    that the equivalence holds for $n$.  Let $n_0=\lfloor\frac n
    2\rfloor, n_1=n-\left(\v(\tilde \t)-\v(\tilde \s)\right)
    $. 

    For any $\s,\s'$, if
    $\Pi_{n+1}(\s)=\Pi_{n+1}(\s')$, then $\H(n,\s)=\H(n,\s')$, since
    then $\tilde s=\tilde s'$.  In particular, with the notation of
    the algorithm, $\Pi_{n_0+1}(\s)=\Pi_{n_0+1}(\tilde \s)$, which
    implies that $\H(n_0,\tilde \s)=\H(n_0,\s)$.  Hence, by recurrence
    assumption, we get that $\mathbf R \cdot \s = \H(n_0, \s) \cdot \s
    \sim \Q_{n_0}(\s) \cdot \s$.

    The definition of $\Q_{n_0}$ implies that $\v(\Q_{n_0}(\s)\cdot s)
    -\v(\s) \le \lfloor n/2 \rfloor$, and by the claim above, we get that
    $\v(\mathbf R \cdot s) -\v(\s) \le \lfloor n/2 \rfloor$.
    This implies that $\v(\mathbf R \cdot s)=\v(\mathbf R \cdot \tilde s)$,
    and,
    since $n'=n- (\v(\mathbf R \cdot \tilde \s) - \v(\tilde \s))$,  that $\Pi_{n'+1}(\mathbf R \cdot \s) =
    \Pi_{n'+1}(\mathbf R \cdot \tilde{\s}) = \tilde{u}$. 

    To continue, we distinguish two cases.  If $\eta(\mathbf R \cdot
    \s) \geq n' + 1$ then the first part of $\Q$ being online shows
    that $\eta(\Pi_{n'+1}(\mathbf R \cdot \s)) = \eta(\tilde u)\ge
    n'+1$.  Hence, in this case, the algorithm returns $\mathbf R$. On
    the other hand, we will prove that in this case, $\Q_{n}(\s) =
    \Q_{n_0}(\s)$; one this is established, this implies that
    $\Q_{n}(\s)\cdot s\sim \mathbf R \cdot \s$, so our correctness
    claim holds in this case. Indeed, $\eta(\mathbf R \cdot
    \s)=\eta(\Q_{n_0}(\s)\cdot s)$ (in view of the equivalence written above,
    and of Lemma~\ref{lem:Euclidean0}), which
    can be rewritten as
    $\v(\Q(\Q_{n_0}(\s)\cdot s)\cdot \Q_{n_0}(\s)\cdot s)-\v(\Q_{n_0}(\s)\cdot s).$
    On the other hand,     $n'=n-(\v(\mathbf R \cdot \tilde \s) - \v(\tilde \s))$
    gives
    $n'=n-\v(\Q_{n_0}(\s)\cdot s)+\v(\s).$
    Hence, $\eta(\mathbf R \cdot
    \s) \geq n' + 1$ means that
    $$\v(\Q(\Q_{n_0}(\s)\cdot s)\cdot \Q_{n_0}(\s)\cdot s) - \v(s) \ge n+1,$$
    which precisely implies that $\Q_{n_0}(\s)=\Q_n(\s)$.

    If $n'+1>\eta(\mathbf R \cdot \s)$, $\Q$ being online leads to the
    equivalence $\Q(\tilde u) \cdot \mathbf R \cdot \s \sim \Q(\mathbf
    R \cdot \s) \cdot \mathbf R \cdot \s$. We claim that the
    right-hand side has valuation $\v(\tilde \t)=\v(\Q(\tilde u)\cdot
    \tilde u).$ Indeed, the proof of Lemma~\ref{lem:QisEuclidean}
    establishes the existence of a matrix $\mathbf K$ such that
    $\Q(\tilde u)=(\Id +\mathbf K) \Q(\mathbf R\cdot s)$; this implies 
    that $\v(\Q(\tilde u)\cdot \tilde u)=\v(\Q(\mathbf R\cdot s)\cdot \tilde u)$.
    On the other hand, the inequality  $n'+1>\eta(\mathbf R \cdot \s)$
    also implies that $\v(\Q(\mathbf R\cdot s)\cdot \tilde u)$
    and $\v(\Q(\mathbf R\cdot s)\cdot \mathbf R \cdot s)$
    are the same, which proves our claim.

    Using $\mathbf R \cdot \s \sim \Q_{n_0}(\s) \cdot \s$, and 
    $\Q$ being Euclidean, we  get $\Q(\mathbf
    R \cdot \s) \cdot \mathbf R \cdot \s  \sim
    \Q(\Q_{n_0}(\s) \cdot \s )\cdot \Q_{n_0}(\s) \cdot \s.$
    We claim that the right-hand side is equivalent
    to $\Q_{n-n_1}(\s)$, which will prove 
    $\Q(\tilde u) \cdot \mathbf R \cdot \s \sim \Q_{n-n_1}(\s)\cdot\s$.
    Indeed, by definition, $\Q_{n-n_1}(\s)$ is the last
    element in the sequence $(\mathbf Q_i\cdot \s)$
    from~\eqref{eq:recursivesequence} having valuation at most
    $\v(\tilde t)$. On the other hand, the previous paragraph proves that  $\Q(\Q_{n_0}(\s) \cdot \s) \cdot
    \Q_{n_0}(\s) \cdot \s$, which belongs to the sequence $(\mathbf
    Q_i\cdot \s)$, has valuation $\v(\tilde \t)$; this is enough to
    conclude, since the valuations $v(\mathbf Q_i\cdot \s)$ increase. 

    Then $\mathbf S = \H(n_1,\tilde \t)$; by recurrence assumption,
    $\H(n_1,.)\sim\Q_{n_1}$. Moreover, $\tilde \t =
    \Pi_{n_1+1}(\Q(\tilde u) \cdot \tilde u) = \Pi_{n_1+1}(\Q(\tilde u)
    \cdot \mathbf R \cdot \s)$, and by construction
    $\H(n_1,\tilde t) = \H(n_1,\Q(\tilde u) \cdot \mathbf R \cdot \s)$,
    so that
    \begin{align*}
      \mathbf S \cdot \Q(\tilde u) \cdot \mathbf R \cdot \s
      &\sim \Q_{n_1}(\Q(\tilde u) \cdot \mathbf R \cdot \s)
        \cdot \Q(\tilde u) \cdot \mathbf R \cdot \s\\
      &\sim \Q_{n_1}(\Q_{n-n_1}(\s)) \cdot \Q_{n-n_1}(\s)\cdot \s.
    \end{align*}
    The definition of
    the sequence $(\Q_n)$ implies that the latter expression is equivalent to $\Q_n(\s) \cdot \s$.

    Finally, since the function $\varphi_{\s,n}$ satisfies
    $\varphi_{\s,n}(\mathbf S \cdot \Q(\tilde \u) \cdot \mathbf R)
    \cdot \s \sim \mathbf S \cdot \Q(\tilde \u) \cdot \mathbf R \cdot
    \s$, we conclude that $\H(n,.) \sim \Q_n$.
\end{proof}

\begin{Lemma}
    \label{lem:pseudoinverse}
    Let $\s=(P,Q)$ be in $\Seq_k$, of degrees at most $d$.  For $n>
    0$, Algorithm~\ref{algo:generichalfgcd} computes
    $\H(n,\s)$ in time
    $\O(\sfM(dn)\log(n) + \sfM(d)n\log(d)).$
\end{Lemma}
\begin{proof}
  For $n=0$, for any $\s$, computing $\H(0,\s)$ takes constant time.
  For a higher value of $n$, remark that the recursive calls are made
  with arguments $n_0,n_1$ that are at most $\lfloor \frac n 2
  \rfloor$: this is clear for $n_0$; for $n_1=n-(\v(\tilde
  t)-\v(\tilde s))$, this is because $\v(\tilde t)-\v(\tilde
  s)=\v(\Q(\tilde u)\cdot \u)-\v(\tilde s)$ must be greater than
  $\lfloor n/2\rfloor$, by definition of $\Q_{n_0}$.

  The matrix $\mathbf R$ has entries of degree at most $d$ (because
  $\tilde s$ does), so computing $\tilde u$ takes time $\O(\sfM(dn))$,
  and its entries have degree $\O(d)$. More precisely, we saw in the
  proof of the previous lemma that $\mathbf R \cdot \s \sim
  \Q_{n_0}(\s) \cdot \s$. Because $\Q_{n_0}$ is obtained by iterating
  $\Q$, the last item in Lemma~\ref{lem:Euclidean} shows that the
  first entry of $\Q_{n_0}(\s) \cdot \s$ has degree less than $d$; as
  pointed out before, this implies the same property for $\mathbf R
  \cdot \s$, and thus for $\tilde u$ (which is a truncation of it).
 
  Computing $\eta$ takes time $\O(\sfM(dn))$ by
  Lemma~\ref{lem:computeEta}, and the same holds for $\tilde t$ by
  Lemma~\ref{lem:Euclidean}, up to an extra term
  $\O(\sfM(d)\log(d))$. In addition, that lemma shows that the entries
  of $\tilde t$ have degree less than that of the first entry of
  $\tilde u$, and thus less than $d$.

  After the last recursive call, it remains to compute $\mathbf M$.
  We first compute $\mathbf Q_1 = \Q(\tilde u) \mod (N_u(P), x^\ell)$ and
  $\mathbf Q_2 = \Q(\tilde u) \mod (N_u(Q), x^\ell)$, where $\ell =
  \v(\mathbf S \cdot \tilde t) - \v(\tilde \s) \leq n$. This can be
  done in time $\O(\sfM(dn))$: $\Q(\tilde u)$ is a product of three
  matrices, called $\mathbf G, \mathbf D,\mathbf N$ in
  Subsection~\ref{ssec:Q}. The first two have polynomial entries of
  degree at most $d$, and can be computed in time $\O(\sfM(dn))$; the
  last one involves a denominator of the form $\Lc(U)^{-1}$, for some
  polynomial $U$ of degree at most $d$. The inverse of $\Lc(U)$ may
  have degree $\Omega(dk)$, but one can directly compute it modulo
  $N_u(P)$ or $N_u(Q)$ using Newton iteration on $x$ in time
  $\O(\sfM(dn))$.

  Then we compute $\mathbf M_1 = \varphi_{\tilde s, \ell}(\mathbf S \cdot
  \mathbf Q_1 \cdot \mathbf R)$ and $\mathbf M_2 = \varphi_{\tilde s,
    \ell}(\mathbf S \cdot \mathbf Q_2 \cdot \mathbf R)$ and we let
  $\mathbf M$ be the concatenation of the first column of $\mathbf
  M_2$ and the second column of $\mathbf M_1$. Thus $\mathbf M$ can be
  computed in time $\O(\sfM(dn))$.

  Overall, the time spent on input $(n,\s)$ is
  $\O(\sfM(dn)+\sfM(d)\log(d))$, plus two recursive calls with
  parameter at most $\lfloor n/2\rfloor$, in degree at most $d$.
  The total is thus $\O(\sfM(dn)\log(n)+\sfM(d)n\log(d))$.
\end{proof}


\subsection{Computing pseudo-inverses}

Given two polynomials $P$ and $Q$ in $\R$, we will use
Algorithm~\ref{algo:generichalfgcd} to compute $U$ and $V$ such that
$UP + VQ = x^t \mod x^{t+1},$ where $t$ is the smallest integer
such that $\left<P,Q\right> \cap \langle x^t\rangle = \langle
x^t\rangle$. In particular, the valuation of the resultant
of $P$ and $Q$ is greater than or equal to~$t$.

\begin{Corollary}
    \label{cor:pseudo-inverse}
    Assume that $P,Q\in \R$ have degree at most $d$, and let $t$ be
    the minimal integer such that $\left<P,Q\right> \cap \langle
    x^t\rangle = \langle x^t\rangle$. It is possible to compute in
    time
    $\O(\sfM(dt)\log (t) + \sfM(d)t\log (d))$ two polynomials $U, V \in \R$ such that
    $   UP + VQ = x^t \bmod x^{t+1},  $
    with $\deg_y(U) < \deg_y(Q)$ and $\deg_y(V)<\deg_y(P)$.
\end{Corollary}
\begin{proof}
  Without loss of generality, assume that $\v(P) \le \v(Q)$.  Suppose
  first that we actually have $\v(P) < \v(Q)$, and define $\s=(P,Q)
  \in \Seq_k$. In the following we let $t' = t - \v(\s)$.

  In particular we have $I_{t'}(\s) = \left<1\right>$ and since for all
  $i<t'$, $I_i(\s)\varsubsetneq I_{t'}(\s)$, the properties of $\Q_{t'}$
  given in Lemma~\ref{lem:Qn} ensure that $\Q_{t'}(\s)\cdot\s = \langle
  s\rangle \cap \langle x^t\rangle = \langle x^t\rangle$.
  Moreover, for any integer $i\geq t'$, we have
  $\Q_i(\s) = \Q_{t'}(\s)$.  Thus, by the
  equivalence property of $\H$ given in Lemma~\ref{lem:generichalfgcd},
  we have for any $i\geq t'$:
  $$\H(i,\s)\cdot \s = \left(\begin{smallmatrix}
  x^t(a+xW)\\ x^{t+1}H\end{smallmatrix}\right),$$
  where $\s=(P,Q)$ and $W, H$ are in
  $\R$ and $a$ is a non-zero constant. Thus it is enough to compute $\H(i,\s)$ for
  any $i \geq t'$ to recover $U$ and $V$ from the first row of the
  matrix $\H(i,\s)$. On the other hand for any $i<t'$, the first
  coordinate of $\Pi_1(\H(i,\s)\cdot \s)$ has a degree greater or equal
  to $1$. Thus we can apply
  Algorithm~\ref{algo:generichalfgcd} to the input $(2^i,s)$, for $i$
  from $1$ to $\lceil \log k \rceil$
  until the first polynomial of
  $\Pi_1(\H(2^i,\s)\cdot \s)$ has degree $0$ in $y$. We will find
  $2^{i_0} \geq t'$ while calling $\H$ at most $\lceil \log t' \rceil$
  times, and this will allow us to conclude.
  
  Suppose now that $\v(P)=\v(Q)$, let $\s=(P,Q)$ and $\s'=\mathbf
  G_\s\cdot \s$, where $\mathbf G$ is the gcd matrix defined in
  Subsection~\ref{ssec:Q}. Then $\s'$ satisfies the assumptions of the
  previous paragraph, we can compute polynomials $U',V'$ such that
  $U'P'+V'Q'=x^t \bmod x^{t+1}$, with $\s'=(P',Q')$. Remark that the
  value of $t$ is indeed the minimal possible one for $\s$ as well,
  since $\mathbf G$ is a unit; note also that the degrees of $P',Q'$
  are $\O(d)$, and thus so are those of $U'$ and $V'$.

  Multiplying $U',V'$ by $\mathbf G$, we obtain polynomials $U'',V''$
  of degree $\O(d)$ such that $U''P+V''Q=x^t \bmod x^{t+1}$. Then,
  define $U=L_Q (L_P U'' \bmod N_u(Q))$ and $V=L_P (L_Q V'' \bmod
  N_u(P))$, with $L_P,L_Q,N_u(P),N_u(Q)$ defined as in
  Lemma~\ref{lemma:varphi}. These polynomials have prescribed degrees,
  can be computed modulo $x^{t+1}$ in time $\O(\sfM(dt))$, and the
  same proof as in Lemma~\ref{lemma:varphi} shows that
  $UP+VQ=x^t \bmod x^{t+1}$.
\end{proof}


\section{The first non-zero coefficient}\label{sec:first-coeff}

Before computing the first $k$ coefficients of the resultant
$R(x)\in\K[x]/\langle x^k\rangle$ of two polynomials $P,Q \in \R$, we focus on
computing the first non-zero coefficient of $R$. The following lemma
will allow us to compute it by recurrence.

\begin{Lemma}
    \label{lem:recurrence}
    Let $M,N \in \R$ and $U\in\U(\R)$ and an integer $t \le k$ be such
    that $MP + NQ = x^t U$ with $\deg_y(N)< \deg_y(P)$, and such that $P$
    and $N$ are
    normal and $\Pi_1(U)=1$. Furthermore,  assume that 
    $\v(P)=\v(N)=0$.

    Denote by $d_P,d_Q,d_M,d_N$ the degrees in $y$ of $P,Q,M,N$
    respectively, by $p_0 \in \K$ the coefficient of $y^{d_P}$ in $P$,
    and by $n_0 \in\K$ the coefficient of $y^{d_N}$ in $N$.  Then,
    there exists $V$ unit in $\K[x]/\langle x^k\rangle$, with $\Pi_1(V)=1$, such
    that:
    $$
    \Res(P,Q) =
    x^{t(d_P-d_N)}(-1)^{d_Nd_P}\frac{p_0^{d_N+d_Q}}{n_0^{d_M+d_P}}V
    \Res(N,M).
    $$
\end{Lemma}
\begin{proof}
    By the multiplication rules of the resultant, 
    \begin{align*}
        \Res(P,N)\Res(P,Q) &= \Res(P,NQ)\\
        \Res(N,P)\Res(N,M) &= \Res(N,MP)
    \end{align*}
    Replacing in the first equality $NQ$ by $NQ + MP = x^tU$, we
    get $\Res(P,NQ) = p_0^{d_N+d_Q-d_U}\Res(P, x^tU)$ (see for example
    \cite{GKZbook94,CLObook05}). Finally, $\Res(P,x^tU) =
    x^{td_P}\Res(P,U)$. Then since $P$ is normal of valuation zero and $\Pi_1(U)=1$,
    there exists  $W=1+x\tilde W \in \K[x]/\langle x^k\rangle$ such that
    $\Res(P,U)=p_0^{d_U}W$. Applying these arguments to
    $\Res(N,MP)$, we conclude:
    \begin{align*}
        \Res(P,N)\Res(P,Q) &= p_0^{d_N+d_Q}x^{td_P}W\\
        \Res(N,P)\Res(N,M) &= n_0^{d_M+d_P}x^{td_N}W'
    \end{align*}
    Note that the symmetry formula for the resultant implies that
    $\Res(P,N) = (-1)^{d_Pd_N}\Res(N,P)$.  Then
    dividing the two equalities, we recover the desired result.
\end{proof}

\begin{Lemma}[First non-zero coefficient]
    Let $d$ be a bound on the degrees in $y$ of $P$ and $Q$. One can
    determine whether $R=\Res(P,Q)$ vanishes in $\K[x]/\langle x^k\rangle$, and if
    not compute its first non-zero coefficient and its valuation in
    $\O(\sfM(dk)\log(k) + \sfM(d)k\log(d))$ arithmetic operations.
\end{Lemma}
\begin{proof}
  Lemma~\ref{lem:normalres} allows us to reduce to the case where $P$
  and $Q$ are normal polynomials; the cost of this reduction is
  $\O(\sfM(dk))$.  Starting from normal $P$ and $Q$, we prove the
  result by induction; the proof actually only uses the fact that one
  polynomial, say $P$, is normal. We use an integer argument $\tau$,
  which gives us a (strict) upper bound on the valuation of the
  resultant; initially, it is set to $k$.


  Dividing by powers of $x$, we can assume that $\v(P)=\v(Q)=0$; the
  upper bound $\tau$ remains valid. We then compute the resultant
  of $\Pi_1(P)$ and $\Pi_1(Q)$ in $\K[y]$, in time
  $\O(\sfM(d)\log(d))$. If it is non-zero, we are
  done. 

  Else, let $t$ be the smallest integer such that $\left<P,Q\right> \cap
  \langle x^t\rangle = \langle x^t\rangle$; hence, $t \le \tau$, but we
  also have $t > 0$. Define function $F(d,n) = \sfM(dn)\log(n) +
  \sfM(d)n\log(d)$. Using Corollary~\ref{cor:pseudo-inverse}, we see
  that there exists a universal constant $c_1$ such that we can
  compute in $c_1 F(d,t)$ operations two polynomials $U$ and $V$ in
  $\R$ of degree less than $d$ such that $UP + VQ = x^tW$ with
  $W\in\U(\R)$, with more precisely $\deg_y(U) < \deg_y(P)$. Since
  $t$ is minimal, this implies $\v(U,V)=0$ and since $\v(P)=0$, this
  implies that $\v(V)=0$.

  If $t=\tau$, we are done. Else, let $A = \Pi_{t+1}(\Lc(V))$,  $N =
  \Pi_{t+1}(N_u(V))$ and $M = \Pi_{t+1}\left( U/A \remy N_u(Q)\right).$
  Since $A$ is a unit, these definitions imply the equality $MP+NQ=x^t
  (1+xY) + Z N_u(Q) P$, for some polynomials $Y$ and $Z$. The degree of
  the left-hand side is less than $\deg_y(P)+\deg_y(Q)$, whereas $N_u(Q) 
  P$ is monic of degree $\deg_y(P)+\deg_y(Q)$. Hence, the previous
  equality shows that $Z N_u(Q) P$ vanishes modulo $x^{t+1}$.
  The assumptions of Lemma~\ref{lem:recurrence} are
  satisfied; we can thus do a recursive call on $N$ and
  $M$, with upper bound $\tau-t$, from which we can recover our output
  using the formula in that lemma.

  In terms of complexity, all calculations giving $A,N,M$ can be 
  done in $c_2 \sfM(dt)$ operations (the only non-trivial point
    is the computation of $1/A \remy N_u(Q)$, which is done by Newton iteration 
  on $x$). Hence, the runtime $G(d,\tau)$ satisfies 
  $G(d,\tau) \le c_0 \sfM(d)\log(d) + c_1 F(d,t)+c_2 \sfM(dt) + G(d,\tau-t).$
  Using the super-linearity of $F$ in $t$ and of $\sfM$, 
  and the definition of $F$, we deduce the overall cost
  $G(d,k) =\O(F(d,k))$.
\end{proof}


\section{A differential equation}\label{sec:diffeq}

Let $P$ and $Q$ be in $\K[x,y]$, and let $R \in
\K[x]$ be their resultant with respect to $y$. We now prove
our main result:
\begin{Theorem}
If $P$ and $Q$ have degree at most $d$ and $\K$ has
  characteristic zero, or at least $k$, one can compute $R\remx x^k$
  using $\O(\sfM(dk)\log (k)+\sfM(d)k\log (d))$ operations in $\K$.
\end{Theorem}

\noindent
{\bf \small First reduction:} If the degree in $x$ of $P$ or $Q$ is
greater than or equal to $k$, let $P_k = P \remx x^k$ and $Q_k
\remx x^k$ and let $d_P,d_Q$ be the degrees in $y$ of $P$ and $Q$
respectively. Then,
$$\Res_{d_P,d_Q}(P_k, Q_k) = \Res(P,Q) \mod x^k,$$ where
$\ResIn_{d_P,d_Q}$ denotes the determinant of the Sylvester matrix
associated to the degrees $(d_P,d_Q)$. If both leading coefficients of
$P$ and $Q$ have a valuation less than $k$, then
$\ResIn_{d_P,d_Q}(P_k, Q_k) = \ResIn(P_k,Q_k)$. If both leading
coefficients have a valuation greater than or equal to $k$ then
$\ResIn(P,Q) = 0 \bmod x^k$.  Finally, if only the leading coefficient
of say $Q$ has a valuation greater or equal to $k$, then we have
$\Res_{d_P,d_Q}(P_k,Q_k) = p_0^{d_Q-\deg_y(Q_k)}\Res(P_k, Q_k)$, where
$p_0$ is the leading coefficient of $P$ in $y$. Thus, in any case, we
can recover the resultant of $P$ and $Q$ modulo $x^k$ from that of
$P_k$ and $Q_k$, in a time that fits in our runtime bound.

\smallskip\noindent {\bf \small Second reduction:} Assume that $P$ and
$Q$ have degree at most $d$ in $y$, with coefficients of degree less
than $k$. Using Lemma~\ref{lem:normalres}, in time $\O(\sfM(dk))$, we
can reduce the problem of computing $\Res(P,Q) \remx x^k$ to a similar
problem with $P$ and $Q$ both monic in $y$, reduced modulo $x^k$, and
with no degree increase in $y$ (that lemma proves the existence of
suitable polynomials in $\R$, so we take their canonical lifts to
$\K[x,y]$). Hence, below, we suppose we are in this case.

With the results of the previous section, we can test if
$R=\Res(P,Q)$ vanishes modulo $x^k$, and if not, find its valuation
$\mu \le k$ and the coefficient $c$ of $x^\mu$, in time
$\O(\sfM(dk)\log k+\sfM(d)k\log d)$. We thus assume that $R \remx x^k$
is non-zero, as otherwise we are done.  The key to our algorithm is
the following differential equation satisfied by $R$ over $\K(x)[y]$;
below, we write $d_P = \deg_y(P)$ and $d_Q=\deg_y(Q)$.
\begin{Lemma}
The following equality holds:
\begin{align*}
 \frac{d R}{d x} & = R \left (
{\rm coeff}\left( \frac 1P \frac{dP}{dx} \frac{dQ}{dy} \remy Q, y^{d_Q-1} \right) \right .\\ 
&+ 
\left .{\rm coeff} \left (\frac 1Q \frac{dQ}{dx} \frac{dP}{dy} \remy P, y^{d_P-1} \right)\right ).
\end{align*}
\end{Lemma}
\begin{proof}
  Let $\mat{A}$ be the Sylvester matrix of $P$ and $Q$ with respect to the variable $y$,
  so that $R=\det(\mat{A})$. Since $R \remx x^k$ is non-zero, $\mat{A}$ is a unit
  over $\K(x)$. 

  Differentiating this equality with respect to $x$, we
  obtain $\frac{d R}{d x} = R \, {\rm trace}\left( \mat{A}^{-1}
  \frac{d \mat{A}}{d x} \right ).$ Since $\mat{A}$ is the matrix of
  the mapping $(F,G) \mapsto F P + G Q$ (in the canonical monomial
  bases),  $\frac{d \mat{A}}{d x}$ represents 
  $(F,G) \mapsto F \frac{dP}{dx} + G \frac{dQ}{dx}$, in the same
  bases.  Similarly, the inverse $\mat{A}^{-1}$ represents the mapping
  $S \mapsto (S/P \remy Q, S/Q \remy P)$. Due to the block structure
  of the matrix $\mat{A}^{-1}   \frac{d \mat{A}}{d x}$, 
  its trace is the trace of the block-diagonal
 operator
  $$ (F,G) \mapsto \left (F \frac 1P \frac{dP}{dx} \remy Q,\ G \frac 1Q \frac{dQ}{dx} \remy P \right ).$$
  This mapping being block-diagonal, its trace is the sum of the traces of its two 
  components, 
  $ F\mapsto  F\frac 1P  \frac{dP}{dx} \remy Q$
and $  G\mapsto   G \frac 1Q \frac{dQ}{dx} \remy P$. To conclude,  remark that these traces
  are respectively
   ${\rm coeff} \left ( \frac 1P \frac{dP}{dx} \frac{dQ}{dy} \remy Q, y^{d_Q-1} \right)$
  and ${\rm coeff} \left (  \frac 1Q \frac{dQ}{dx} \frac{dP}{dy} \remy P, y^{d_P-1} \right). $
 \end{proof}

Using Corollary~\ref{cor:pseudo-inverse},
we can compute in time $\O(\sfM(dk)\log (k) + \sfM(d)k\log (d))$
we know two cofactors $U$ and $V$ in $\K[x,y]$, of degree less than $k$ in $x$ such that
$\deg(U,y) < d_Q$, $\deg(V,y) < d_P$ and
$U P + V Q = x^t \bmod x^{t+1},$ with $t$ chosen minimal, so that $t
\le \mu$. Then, $t \le k$, since otherwise $R = 0 \mod x^k$.

From this, we deduce $U',V'$ with the same degree constraints on $y$
and degree less than $2t+k-1$ in $x$, such that $$ U'P + V'Q = x^t
\bmod x^{2t+k-1}:$$ compute $W$ such that $UP+VQ = x^t(1+xW) \mod
x^{2t+k-1}$, then the inverses $A$ (resp.\ $B$) of $1+xW$ modulo
$\langle Q,x^{2t+k-1}\rangle$ (resp.\ modulo $\langle
P,x^{2t+k-1}\rangle$), and let $U' = UA \remy Q$ and $V' = VB \remy
P$. The only non-trivial point is the inversions, which are done by
Newton iteration with respect to $x$; overall, the cost of this step
is $\O(\sfM(dk))$.

The defining equality for $U'$ and $V'$ can be rewritten as $$U'P +
V'Q = x^t(1+x^{t+k-1}S),$$ for some $S$ in $\K[x,y]$. Their
degree constraints then show that the inverse of $P$ modulo $Q$ is
$U'/(x^t(1+x^{t+k-1}S)) \remy Q$; similarly, the inverse of $Q$ modulo
$P$ is $V'/(x^t(1+x^{t+k-1}S)) \remy P$. This further implies that $x^t
\frac {dR}{dx}$ is equal to
\begin{align*}
  R \left (\, {\rm coeff} \left  ( \frac {U'}{1+x^{t+k-1}S} \frac{dP}{dx} \frac{dQ}{dy} \remy Q, y^{d_Q-1} \right) \right. \\
  + \left . {\rm coeff}  \left ( \frac {V'}{1+x^{t+k-1}S} \frac{dQ}{dx} \frac{dP}{dy} \remy P, y^{d_P-1} \right ) \right ).
\end{align*}
Taking this equality modulo $x^{t+k-1}$, we obtain a relation of the form
$x^t \frac {dR}{dx} = R F \bmod x^{t+k-1},$
with
\begin{align*}
  F &=   \left .\, {\rm coeff} \left  ( U' \frac{dP}{dx} \frac{dQ}{dy} \remy Q, y^{d_Q-1} \right) \right. \\
  & + \left . {\rm coeff}  \left (V' \frac{dQ}{dx} \frac{dP}{dy} \remy P, y^{d_P-1} \right ) \right . \bmod x^{t+k-1}.
\end{align*}
Because $P$ and $Q$ are both monic in $y$, once $U'$ and $V'$ are
known, we can compute $F$ using $\O(\sfM( (t+k) d ))$ operations in
$\K$, which is $\O(\sfM(kd ))$.

Recall that $R$ is has the form $c x^\mu + \cdots$, for
some $\mu < k$, so that $\frac {dR}{dx}$ has the form $\mu c x^{\mu-1} +
\cdots$. Thus, the Laurent series $\frac 1R \frac {dR}{dx}$ has valuation
at least $-1$, so that $F$ has valuation at least $t-1$. Dividing by $x^{t-1}$
on both sides, we obtain 
$x \frac {dR}{dx} = R \tilde F \bmod x^{k},$ with $\tilde F = F /
x^{t-1}$.  From now on, let us assume that the characteristic $p$ of
the base field is at least equal to $k$, or zero. Then, this relation
determines $R \bmod x^k$ up to a constant factor, and knowing the
initial condition $c$ allows us to deduce $R\bmod x^k$
unambiguously. Given $\tilde F$, this is done by means
of~\cite[Theorem~2]{BoChLeSaSc12}, which allows us to compute $R\bmod
x^k$ in $\O(\sfM(k))$ operations in $\K$. Summing all the costs seen
so far concludes the proof of our theorem.

\smallskip\noindent{\bf \small Acknowledgements.} The authors
sincerely thank the reviewers for their careful reading and
suggestions. Moroz is supported by project Singcast
(ANR-13-JS02-0006); Schost is supported by NSERC.

{\small
\bibliographystyle{abbrv}
\bibliography{resultant_series}
}

\end{document}